\documentclass[a4paper,twocolumn,11pt]{quantumarticle}
\pdfoutput=1

\usepackage{amsmath,amsthm,amssymb, bm}
\usepackage{cancel} 
\usepackage[normalem]{ulem} 
\usepackage{soul}
\usepackage{graphicx}
\usepackage{physics}
\usepackage{bbm}
\usepackage{xcolor} 
\usepackage[hyperindex,breaklinks,colorlinks]{hyperref}

\usepackage[numbers,sort&compress]{natbib}

\usepackage{ragged2e}

\newtheorem{lemma}{Lemma}

\begin{document}
\title{Background Independence and Quantum Causal Structure}
\author{L. Parker}
\email{l.parker@uq.net.au}
\affiliation{Centre for Engineered Quantum Systems, School of Mathematics and Physics, University of Queensland, QLD 4072 Australia}
\author{F. Costa}
\email{f.costa@uq.edu.au}
\affiliation{Centre for Engineered Quantum Systems, School of Mathematics and Physics, University of Queensland, QLD 4072 Australia}

\maketitle

\begin{abstract}
One of the key ways in which quantum mechanics differs from general relativity is that it requires a fixed background reference frame for spacetime. In fact, this appears to be one of the main conceptual obstacles to uniting the two theories. Additionally, a combination of the two theories is expected to yield `indefinite' causal structures. In this paper, we present a background-independent formulation of the process matrix formalism---a form of quantum mechanics that allows for indefinite causal structure---while retaining operationally well-defined measurement statistics. We do this by imposing that the probabilities arising in the formalism---which we ascribe to measurement outcomes across the points of a discrete spacetime---be invariant under permutations of spacetime points. We find (a) that one still obtains nontrivial, indefinite causal structures with background independence, (b) that we lose the idea of local operations in distinct laboratories, but can recover it by encoding a reference frame into the physical states of our system, and (c) that permutation invariance imposes surprising symmetry constraints that, although formally similar to a superselection rule, cannot be interpreted as such.
\end{abstract}

\section{Introduction}
In a quantum theory of gravity it is expected that spacetime itself will be quantised, giving rise to indefinite, or `quantum', causal structures \cite{Butterfield2001, Hardy2005}. The process matrix formalism was developed to describe these causal structures \cite{oreshkov:2012}. In fact, it  describes the most general causal relations between a finite set of regions, or `parties', compatible with the local validity of quantum mechanics in each region. However, the framework relies on an \emph{a priori} labelling of the parties, which tacitly presupposes the existence of a background reference frame. This is in conflict with the background independence of general relativity, which associates no absolute meaning to individual spacetime points or regions \cite{rovelli:1991, sep-hole}. Incorporating background independence into the quantum formalism is in fact one of the main challenges in the development of a theory of quantum gravity \cite{Ashtekar2004, Smolin2006}.

In order to represent a viable approach to quantum gravity, the process matrix formalism should be able to describe indefinite causal structures without reference to a fixed background. Here, we show how this can be done within a toy model, where we treat a process matrix as a particular configuration of a discretised spacetime, with laboratories that correspond to the discrete units of that spacetime. A process matrix will be background independent if it is invariant under any arbitrary permutation of `laboratories' or points of spacetime.
	
In this paper, we introduce background independent processes and describe some of their properties. First, we note that indefinite causal structures still arise in permutation-invariant processes. We show that imposing permutation invariance results in the loss of a distinction between spacetime points. As in general relativity, one recovers a distinction between spacetime points by using a material reference frame (a reference frame made up of physically observable systems, the `rods and clocks' picture).

Finally, we discuss the symmetry properties of permutation-invariant processes. We expect, in analogy to symmetries in ordinary quantum mechanics, that permutation invariance should give rise to a superselection rule for some `charge'---a property that can only take definite values, with coherence between different charge values forbidden (but classical probabilistic mixtures allowed). However, we find that, while indeed coherence between charge sectors is suppressed, processes cannot in general be understood as mixtures of different charge values, as the definite-charge components of the mixture are not valid processes. We show explicitly why this occurs in the case of a bipartite qubit process (where a `qubit' process is just one with two-dimensional local Hilbert spaces). We also present a partial proof generalising that result to any process matrix dwelling in the symmetric or antisymmetric subspaces of the symmetric group. Our results suggest that no invariant processes with a definite charge may exist, although more work will be needed to substantiate this conjecture. The breakdown in the association between symmetries and superselection rules may indicate that background independence in quantum mechanics cannot be interpreted analogously to other known symmetries of nature.

\section{Quantum processes in spacetime}
We will first review the process matrix formalism in its most common presentation---using a language borrowed from quantum information. Afterwards, we will discuss how this formalism can be used as a model of discretised spacetime.
\subsection{The process matrix formalism}\label{pm-formalism}

The process matrix formalism is a framework for quantum mechanics that does not assume any global background causal structure, just that quantum mechanics is obeyed locally. Conceptually, it extends quantum mechanics in a similar way in which, relaxing global Lorentz invariance, one can extend special relativity to general relativity. Relaxing the assumption of causal structure allows one to obtain new, `indefinite' causal relations that are incompatible with a fixed ordering of events. Relationships of this type have been observed in the laboratory \cite{procopio:2015, rubino2017, rubino2022experimental, goswami:2018, goswami2018communicating, wei2019experimental, Taddei2021, guo2020experimental, Rubino2020}, where the lack of causal order arises from temporally delocalised events, rather than from a quantum spacetime \cite{Guerin2018, Oreshkov2019timedelocalized}. Much of the experimental interest derives from the applications of indefinite causal relations to computation and communication \cite{chiribella09, chiribella12, colnagh:2012, araujo14, feixquantum2015, Guerin2016, Ebler2018, Salek2018, Gupta2019}.  Here, we briefly describe the aspects of the process formalism that are relevant to this work. For more details, see references \cite{oreshkov:2012,shrapnel:2018,araujo:2015}.

The simplest way to think of process matrices is as follows. Consider a system of $N$ laboratories. Each laboratory is occupied by an experimenter capable of performing all of the preparations, operations, and measurements compatible with the standard measurement formalism of quantum mechanics. Formally, this means that each experimenter has the ability to perform a \textit{quantum instrument}---a set $\mathcal{I}^x=\{\mathcal{M}_i^X\}_{i=1}^n$ of completely positive (CP) maps that sum to a completely positive and trace preserving (CPTP) map. The superscript $X$ denotes that the maps $\mathcal{M}^X:\mathcal{L}\left(\mathcal{H}^{X_I}\right)\to \mathcal{L}\left(\mathcal{H}^{X_O}\right)$ act on laboratory $X$. The Hilbert spaces $\mathcal{H}^{X_I}$, $\mathcal{H}^{X_O}$, respectively represent the incoming and outgoing state-spaces of laboratory $X$, with $\mathcal{L}\left(\mathcal{H}\right)$ denoting the linear operators on $\mathcal{H}$.

Consider the case where we have two parties, Alice and Bob, who respectively have access to instruments $\mathcal{I}^{A}=\{\mathcal{M}^A_i\}$ and $\mathcal{I}^B=\{\mathcal{N}^B_j\}$. The probability that Alice and Bob realise a particular combination of operations $\mathcal{M}^{A}_{i},\mathcal{N}^{B}_{j}$ is given by some probability distribution $P(\mathcal{M}^{A}_{i},\mathcal{N}^{B}_{j})$. To be consistent with quantum mechanics, $P$ must be a \textit{multilinear map} \cite{shrapnel:2018}. The \textit{Choi-Jamio{\l}kowski isomorphism} \cite{jamio72, Choi1975} allows us to represent these operations by sending CP maps $\mathcal{M}^{X}$ to positive semidefinite linear operators $M^{X_IX_O}_i:=[\mathcal{I}\otimes\mathcal{M}^X_i(\ketbra{\phi^+}{\phi^+})]^T\in \mathcal{L}(\mathcal{H}^{X_I}\otimes \mathcal{H}^{X_O})$, where $\ket{\phi^+}=\sum_i\ket{i}^{X_I}\otimes\ket{i}^{X_I}$ is a non-normalised maximally entangled state and $^T$ denotes transposition in the computational basis. These operators act over an input Hilbert space $X_I$ and an output Hilbert space $X_O$. In this representation, the trace preserving condition reads $\Tr_{X_O}M^{X_IX_O}=\mathbbm{1}^{X_I}$; this means that, for a set of maps that form an instrument, we must have $\Tr_{X_O}[\sum_iM^{X_IX_O}_i]=\mathbbm{1}^{X_I}$.

\begin{figure}[ht]
\centering
\includegraphics[width=\linewidth]{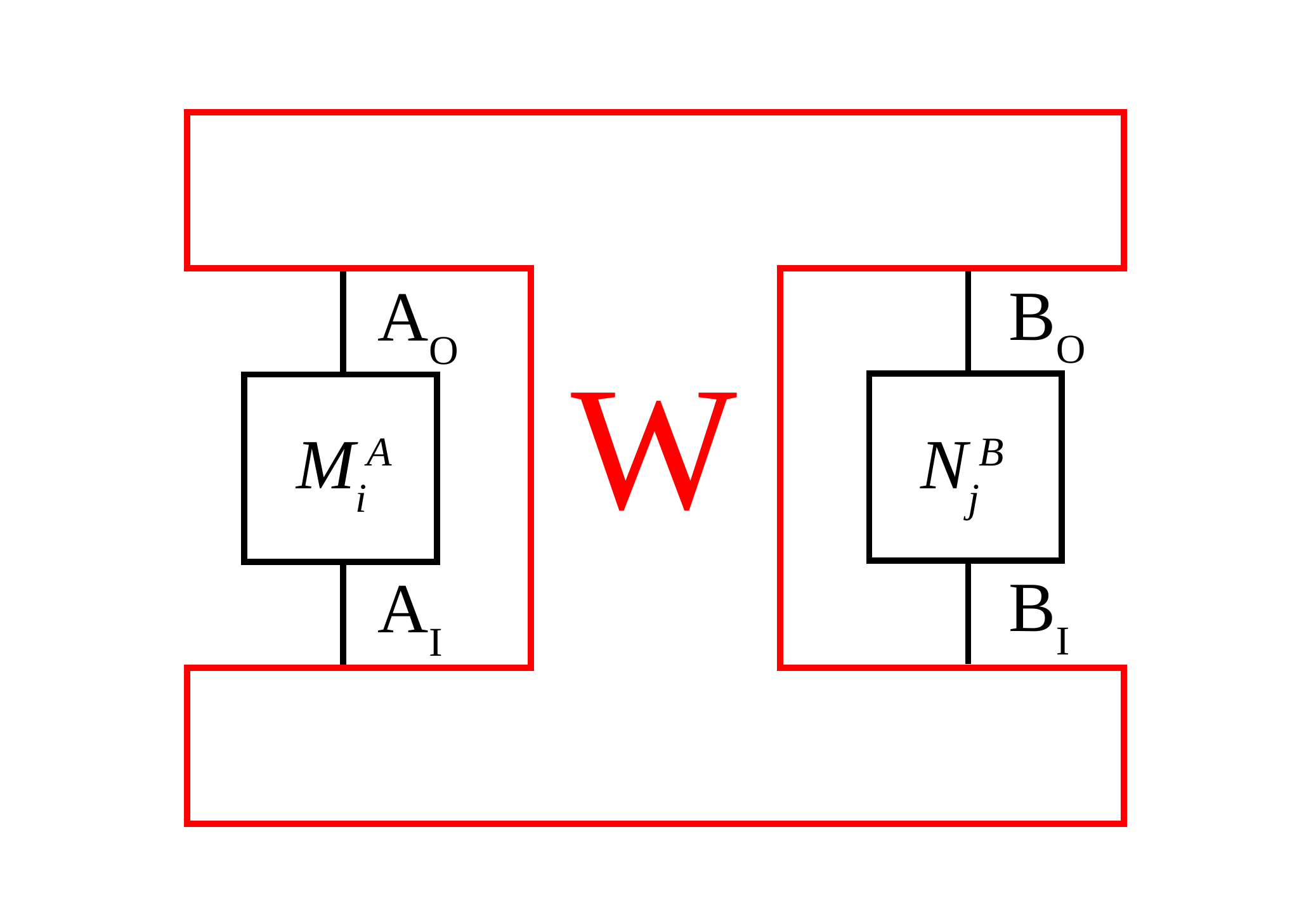}
\caption{A graphical depiction of a process matrix $W=W^{AB}$ with two laboratories, $A$ and $B$. Operations $M_i^A$ and $N_j^B$ are performed on the systems entering laboratories $A$ and $B$, respectively. The probability of $M_i^A$ and $N_j^B$ being realised out of their respective instruments is given by $P(M_i^A,N_j^B) = \Tr[W^{A}\cdot(M^{A}_{i}\otimes N^{B}_{j})]$ (Eq.~\eqref{eq:born-rule}. Note that we have shortened the laboratory labellings from those used in Eq.~\eqref{eq:born-rule}, so that, in the superscripts, $A=A_IA_O$ and $B=B_IB_O$.}
\end{figure}

Our complete list of probabilities $P$ now becomes a multilinear map over linear operators. This map is equivalent to \cite[prop.~2.38]{heinosaari}
\begin{align}
	&P(M^{A_IA_O}_{i}\otimes N^{B_IB_O}_{j})\nonumber\\
	&=\Tr[W^{A_IA_OB_IB_O}\cdot(M^{A_IA_O}_{i}\otimes N^{B_IB_O}_{j})] \label{eq:born-rule}
,\end{align}
for some linear operator $W^{A_IA_OB_IB_O}\in \mathcal{L}(\mathcal{H}^{A_IA_OB_IB_O})$.
	
$W^{A_IA_OB_IB_O}$ is called a \textit{process matrix}, and is the generalisation of a joint quantum state (from the point of view of a probability measure) to correlations that can be spacelike, timelike, or neither---those with indefinite causal structure. Process matrices must satisfy the constraints
\begin{gather}
		W^{A_IA_OB_IB_O}\geq 0 , \label{eq:pm-positive}\\ 
		\Tr[W^{A_IA_OB_IB_O}\cdot(M^{A_IA_O}N^{B_IB_O})]=1, 
		\label{eq:pm-normalised} \\
		\nonumber
		\forall\, M,N\geq0 , \\
		\nonumber
		\Tr_{A_O}[M^{A_IA_O}]=\mathbbm{1}^{A_I}, \Tr_{B_O}[N^{B_IB_O}]=\mathbbm{1}^{B_I},
\end{gather}
which ensure that probabilities are nonnegative and sum to one. We have left out tensor product symbols for convenience, and will continue to do so where it is clear.

In a \textit{Hilbert-Schmidt} basis, i.e., a basis $\{\sigma^X_i\}$ of $\mathcal{L}\left(\mathcal{H}\right)$ satisfying $\sigma_0^X=\mathbbm{1}^X$, $\Tr[\sigma^X_i \sigma^X_j]=d_X \delta_{ij}$, ($d_X:=\dim(\mathcal{H}^X)$) and $\Tr[\sigma^X_i]=0$ for $i>0$, a process matrix can be represented as
\begin{equation}
	W^{A_IA_OB_IB_O}=\sum_{ijkl}w_{ijkl}\sigma_i^{A_I}\sigma_j^{A_O} \sigma_k^{B_I} \sigma_l^{B_O}
,\end{equation}
where $w_{ijkl} \in \mathbb{R}$ since $W^{A_IA_OB_IB_O}$ is hermitian. 

The probability normalisation requirement forbids certain Hilbert-Schmidt terms from appearing in the decomposition of an allowed process.
We call terms with identity on all outputs except $X$ type $X$ process terms, all outputs except $X,Y$ terms of type $XY$ etc. Forbidden bipartite process terms are terms of the form $A_O$, $B_O$, $A_OB_O$, $A_IA_OB_O$, $A_OB_IB_O$, and $A_IA_OB_IB_O$. The process constraint, Eq.~\eqref{eq:pm-normalised}, requires that $\Tr[W \sigma]=0$ for any such terms $\sigma$. Thus, we effectively have linear constraints on the matrix elements of allowed processes.

One consequence of not making \emph{a priori} assumptions about the causal structure is the appearance of novel types of causal order that cannot be expressed in the standard formalism of quantum mechanics. Process matrices can be causally ordered, which corresponds to the familiar situation where A comes before B comes before C, or they can be causally \textit{separable}, convex combinations of processes that have different causal orders such as `A before B' and `B before A', representing classical ignorance of causal order. One novel aspect of the process matrix formalism is that one can also have indefinite causal order, where it does not make sense to say that `A is before B' or vice versa: there are signalling correlations from A to B and also from B to A, which cannot be interpreted as classical ignorance.

Throughout this section we have only discussed bipartite processes, for simplicity. Everything we have discussed generalises straightforwardly to an arbitrary number of parties. We refer the reader to references \cite{oreshkov:2012,shrapnel:2018,araujo:2015, abbott:2016,Abbott2017genuinely,wechs:2019} for a more complete discussion.

\subsection{Process matrices as toy spacetime configurations}
The crux of this paper is in the fact that we can use process matrices as toy models of spacetime. It is worth going into some detail about what exactly this means.

Most of the literature on quantum causal structures, in particular that concerned with experiments or practical applications, regards a process matrix as a representation of a finite number of finite-size `regions', which can be thought of as physical laboratories in which measurements are made on some quantum state. This is the interpretation that has been used in our discussion of the process matrix formalism so far.

It is now time to move beyond this interpretation. For the rest of this paper, we will treat a given process matrix as a toy model of a discretised spacetime. This is consistent with several approaches to quantum gravity, which stipulate that spacetime is discrete at a fundamental level \cite{Myrheim1978, Bombelli1987, Hooft1996, Loll1998, Dowker2006}. However, here we do not consider any explicit model, but simply a generic framework that allows us to avoid the complications associated with continuous spacetime.

Within this framework, each of what we have been referring to as `laboratories' now corresponds to a fundamental, indivisible unit of our discrete spacetime ---in other words, to a single spacetime point. Each indivisible point can be understood as the theoretical maximal spatio-temporal resolution of a measurement apparatus. Occasionally, we will still refer to the elementary spacetime points as `laboratories' or `parties', in keeping with the terminology in the literature, but it must be remembered that this does not entail the existence of a physical laboratory at each point.

We can make the connection with the more ordinary conception of spacetime more explicit. This is usually defined as a differentiable manifold, consisting of an uncountable set $\mathfrak{M}$ that is locally diffeomorphic to Minkowski space. 
Any classical physical property, including the metric, is represented in the form of fields, namely as functions $\phi(x)$ defined at each spacetime point $x\in\mathfrak{M}$.
In a quantum theory, the fields at each point are not uniquely defined: a measurement at $x$ can yield a (possibly uncountable) set of outcomes $\left\{\phi_i(x)\right\}_i$, where $i$ labels the outcomes. A measurement at the point $x$ is clearly an idealisation, with any physical measurement extending to a finite region, but the basic operational meaning of a spacetime point is still the (limit of) the smallest possible region where a measurement can be made.
\begin{figure}[ht]
\centering
\includegraphics[width=\linewidth]{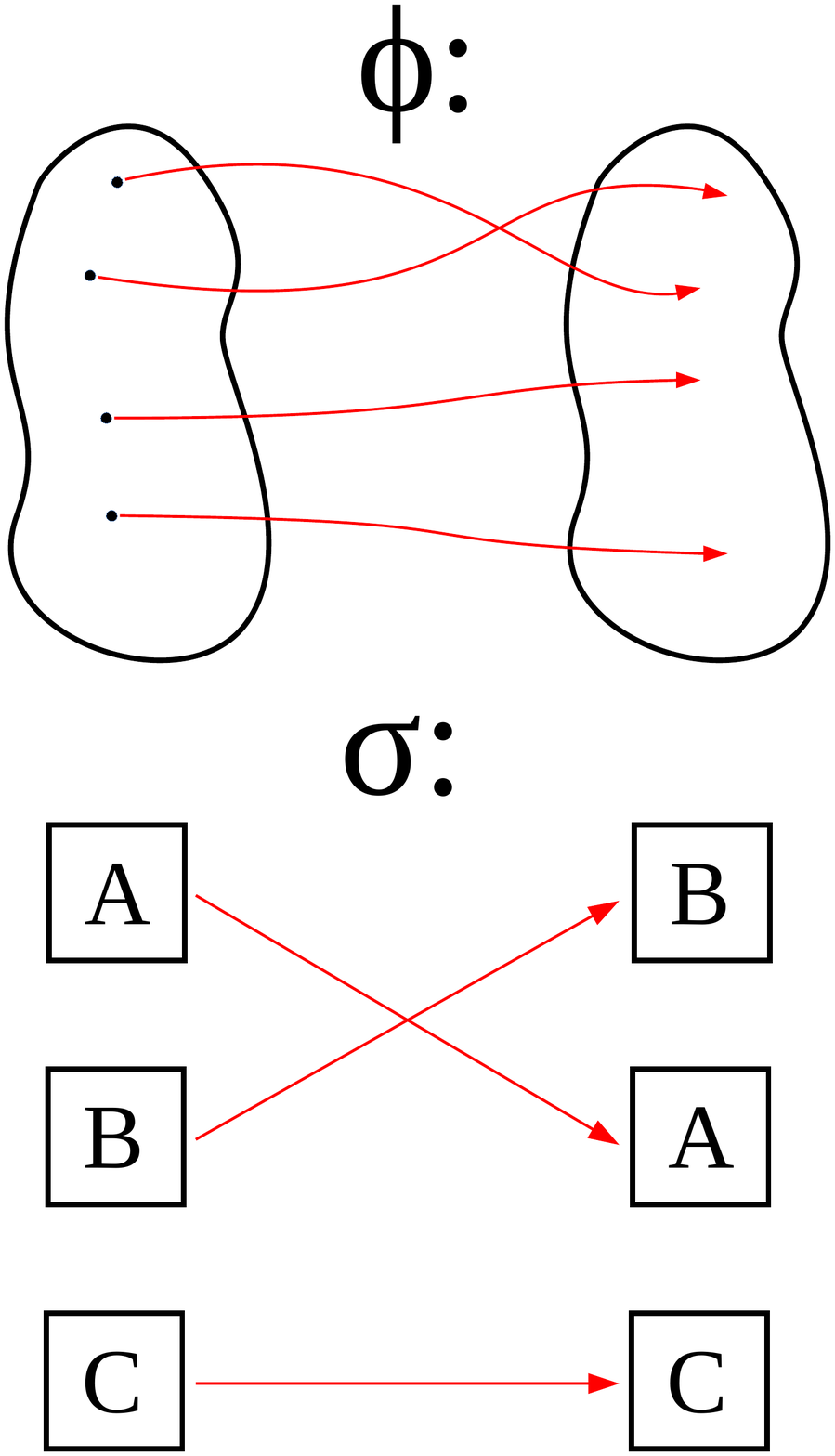}
\caption{A visual representation of a diffeomorphisms $\phi$ (smooth relabelling of points in a manifold) and a permutation $\sigma$ (relabelling of elements in a finite set). A permutation is the discrete analogue of a diffeomorphism.}
\end{figure}

Our model simply replaces the uncountable set $\mathfrak{M}$ with a finite set of spacetime points, which we label $x=1,\dots,n$. (In section \ref{pm-formalism} above, we used  $x=A,B,\dots$ Other labellings will be used below according to convenience). As before, physical quantities can be defined at each point. Importantly, pairs of points in spacetime can be timelike separated (or even have undefined causal relations), meaning that a measurement at one point can influence measurements at different points. An instrument is the most general representation of a measurement (or operation), including how it affects the measured system. An individual CP map $\mathcal{M}^1_{i_1}$ within an instrument represents the measurement of a physical quantity at point $1$, yielding outcome $i_1$, with a transformation of the system described by $\mathcal{M}^1_{i_1}$.
The Born rule for processes---Eq.~\eqref{eq:born-rule} extended to $n$ points---governs the probability distribution $P(i_1,\dots,i_n)$ to observe outcomes $i_1,\dots, i_n$ at spacetime points $1,\dots,n$. Note that it is not necessary to imagine an observer acting at each point in spacetime: We can have points where no operation is performed, corresponding to an instrument with a single element---the identity map.

The process matrix contains all the information relevant for evaluating the probabilities for arbitrary operations at any spacetime point. In a traditional model, such information would be encoded in an initial state, in the dynamical laws describing the propagation of systems across spacetime, and in the metric encoding the causal connections between points. In a successful formulation of quantum gravity, the process matrix should unify
such information in a way that does not depend on a background causal structure.

\section{Why permutation invariance?}

By thinking of a process matrix as representing a particular configuration of a discrete spacetime, we can make an analogy between background independence in the process matrix formalism and background independence in general relativity.

In general relativity, background independence is a consequence of the fact that observable quantities must be invariant under any arbitrary coordinate transformation\footnote{Note that this is a much stronger requirement than simply stating that we can use arbitrary coordinate systems. Indeed, most theories include observables that are not invariant under coordinate transformations. For example, in relativistic quantum field theory, we can measure the value of a field at a point $x$ or at a different point $y$, and these are different observables that can yield different outcomes. Such observables are not allowed in general relativity, as coordinate transformations act both on the fields and on the background spacetime.}. Formally, these transformations are smooth, invertible mappings from a manifold to itself, and are called \textit{diffeomorphisms}. In a discrete spacetime model, the smoothness requirement does not apply. Hence, a discrete coordinate transformation is any invertible map from the set of spacetime points onto itself. As our spacetime toy model contains only a finite set of points, its coordinate relabellings are defined as arbitrary permutations. Therefore, our task is to characterise invariance of observable quantities under permutations in the process matrix formalism.

In the process matrix formalism, the statistical properties of observables are given by e.g.~Eq.~\eqref{eq:born-rule}, the Born-rule generalisation for processes. In general, Eq.~\eqref{eq:born-rule} generates a multipartite probability distribution $P(i_1,...,i_n)$ given a particular process $W$, where $P(i_1,...,i_n)$ denotes the probability that measurement at point $1$ obtains outcome $i_1$, measurement at point $2$ obtains outcome $i_2$, etc. Although $P$ does not assume any causal structure, it does in general assume that it is possible to distinguish between and label the different points. Operationally, this implies the existence of some background reference frame, which allows one to determine that outcome $i_1$  corresponds to party 1, outcome $i_2$ to party 2, and so on. In other words, the labels of the parties form a (discrete) system of coordinates. In a background-independent theory, no background reference frame is available, so it is not possible to identify a point in spacetime with a coordinate. Therefore, the theory should be invariant under relabelling.

As discussed above, in a discrete spacetime model relabellings are simply permutations. As a consequence, in order for observable quantities to be invariant under relabellings, the probability distribution $P$ must be invariant under permutations of the parties, i.e.~$P(i_1,i_2,...,i_n)=P(i_{g(1)},i_{g(2)},...,i_{g(n)}),$ for all permutations $g\in S_n$. Permutation invariance as a discrete analogue of diffeomorphism invariance is also discussed, for example, in Ref.~\cite{Arrighi2020}.

Invariance under permutations is a particular case of invariance under an arbitrary symmetry group. A general framework for dealing with this has been developed in Ref.~\cite{bartlett:2007}. Although this framework deals with Lie groups rather than finite groups (such as permutations), the main results, which we will use below, also hold for finite groups.

First, we must introduce a mathematical representation for permutations, which we will use throughout the paper. Just as the group of diffeomorphisms are represented by a (continuous) diffeomorphism group, the (finite) group of \textit{permutations} of a set of $n$ elements is known as the \textit{symmetric group} and is denoted $S_n$. 

We define the \textit{representation} of the symmetric group $S_n$ on the space of $n$-party process matrices as the map from elements $g\in S_n$ to operators $U_g$ such that $U_gWU^\dag_g$ performs a permutation on the laboratories. For example, the action of the `swap' permutation $U_{AB}$ on a bipartite process in the Hilbert-Schmidt basis is
\begin{multline}
	U_{AB}\left(\sum_{ijkl}\alpha_{ijkl}\sigma^{A_I}_i\sigma^{A_O}_j\sigma^{B_I}_k\sigma^{B_O}_l\right)U^\dag_{AB} \\
	=\sum_{ijkl}\alpha_{ijkl}\sigma^{A_I}_k\sigma^{A_O}_l\sigma^{B_I}_i\sigma^{B_O}_j.
	\label{eq:U-def}
\end{multline}
Note that input and output spaces are always swapped together. In order to make permutations well-defined, we assume that the input spaces of all laboratories have equal dimensions, and similarly all output spaces. It is not difficult to generalise, but we will not do so here \cite{thesis}.

We say that a linear operator {$A$} (which can be a process matrix or a measurement, or more generally even a quantum state or POVM element) is \textit{permutation invariant} if it is unchanged by the action of any permutation, i.e.~$U_gA{U_g}^\dag=A,\ \forall g\in S_n$. Equivalently, $A$ is permutation invariant if $\mathcal{G}[A]=A$, where $\mathcal{G}$ is the \textit{twirl operation},
\begin{equation}
	\mathcal{G}[W]:= \frac{1}{n!}\sum_{g\in S_n}U_gW{U_g}^\dag.
\end{equation}

\section{Process matrices without spacetime events.}

Now we can formalise the ideas we introduced in the previous section. Our overarching goal is to develop a framework for processes in which measurement statistics are permutation-invariant, so that the processes are background-independent. We find that there are different ways to achieve this.

One way to ensure permutation invariant statistics is to restrict measurements to be permutation invariant. As we will discuss below, invariant measurements are not necessarily of the product form. Therefore, we will denote a general measurement operator across all labs with a single symbol $M_{\bm{i}}$, where $M_{\bm{i}}\equiv M_{i_1}^AM_{i_2}^B...$ for the particular case of product measurements and  $\bm{i}\equiv (i_1, i_2,\dots)$ denotes the collection of outcomes (which are still associated with the distinct labs/spacetime points). If $M_{\bm{i}}$ is invariant, we have $P(i_{g(1)},i_{g(2)},...,i_{g(n)}) = \Tr[W U_g M_{\bm{i}} U_g^\dag] = \Tr[W M_{\bm{i}}] = P(i_{1},i_{2},...,i_{n})$, as required.  However, $\Tr[W U_g M_{\bm{i}} U_g^\dag]=\Tr[U_g^\dag W U_g M_{\bm{i}}]$, so $W=\mathcal{G}[W]$ also implies that measurement statistics are invariant, even if $M_{\bm{i}}\neq \mathcal{G}[M_{\bm{i}}]$. Finally, the measurement statistics will be permutation-invariant if both $W$ and $M_{\bm{i}}$ are permutation-invariant. Although they all represent the same physics, each choice is suggestive of a different way to describe permutation invariance.

\textit{If measurement operators are permutation-invariant}, but not the processes themselves, then we might think of processes as being described relative to some fixed background that we cannot access, so that we are restricted to using permutation-invariant measurements.

\textit{If process matrices are permutation-invariant}, but not the measurement operators, then we might instead say that we are making measurements relative to some background reference frame, but that what we observe is permutation-invariant---choosing a different reference frame will give us the same statistics.

\textit{If both process and measurements are permutation-invariant}, then we have totally abandoned the concept of a fixed, absolute reference frame. As in General Relativity, any reference frame must be defined using the available degrees of freedom of the system, the `rods and clocks' picture. The ordering with respect to which we describe our system is not a physically relevant aspect of the theory, just like the description of spacetime relative to a particular choice of coordinates in General Relativity. In other words, the original ordering, implied by the tensor product of local Hilbert spaces, merely represents a fictitious, rather than physical, reference frame, and permutations only represent transformations relative to such a fictitious reference frame.

The three pictures we just described are operationally equivalent, and it is sometimes helpful to work in one picture over another. For example, by requiring that process matrices $W$ are permutation-invariant, $W=\mathcal{G}[W]$, we can see that even with as strict a constraint as permutation-invariance one still obtains nontrivial behaviour of processes. Consider the process matrix
\begin{multline}
	W = \frac{1}{4}\big[\mathbbm{1}^{\otimes4} + a'_0\sigma_z \mathbbm{1}\sigma_z\mathbbm{1} - a'_1(\sigma_z\mathbbm{111} + \mathbbm{11}\sigma_z\mathbbm{1}) \\
	-a'_2(\sigma_z\mathbbm{11}\sigma_z+\mathbbm{1}\sigma_z\sigma_z\mathbbm{1}) +a'_3(\sigma_z\mathbbm{1}\sigma_z\sigma_z+\sigma_z\sigma_z\sigma_z\mathbbm{1}) \\
	+a'_4(\sigma_z\mathbbm{1}\sigma_x\sigma_x-\sigma_z\mathbbm{1}\sigma_y\sigma_y+\sigma_x\sigma_x\sigma_z\mathbbm{1}-\sigma_y\sigma_y\sigma_z\mathbbm{1})\big],
\label{eq:bipartite_causal_indefinite}
\end{multline}
which was presented in Ref.~\cite{branciard:2015}, with coefficients $\vb{a'} = (0.0390, 0.3355, 0.2451, 0.4291, 0.2097)$. In Eq.~\eqref{eq:bipartite_causal_indefinite} we have omitted labels, so that $ABCD=A^{A_I}B^{A_O}C^{B_I}D^{B_O}$. As shown in \cite{branciard:2015}, this process can violate a causal inequality---a device-independent test for indefinite causal order similar to a Bell inequality. Therefore, it represents a minimal example of permutation-invariant process with no definite causal order (in the sense that it involves only two parties, each with a single-qubit system).

Another important consequence that arises is that permutation-invariant process matrices cannot be causally ordered. We will demonstrate this with an example. Consider, for simplicity, the framework in which we impose permutation invariance on processes but not on measurements. Take a  process representing a state $\varrho$ prepared in laboratory $A$ and sent to laboratory $B$ through a channel $T$, $W=\varrho^{A_I}T^{A_OB_I}\mathbbm{1}^{B_O}$. This is not permutation-invariant: $A$ can signal to $B$, but $B$ {cannot signal} to $A$. We can make the process invariant by taking the mixture $W^{\mathrm{inv}}=(\varrho^{A_I}T^{A_OB_I}\mathbbm{1}^{B_O}+\varrho^{B_I}T^{B_OA_I}\mathbbm{1}^{A_O})/2$, noting the change in superscripts in the second term. $W^{\mathrm{inv}}$ may be permutation-invariant, but we have lost the ability to determine whether the state $\varrho$ was prepared in laboratory $A$ and then sent to laboratory $B$, or the reverse: we cannot perform a measurement that will tell us \textit{where} the state preparation occurred. We have lost a way to label laboratories or, equivalently, a definition of spacetime points---we no longer have a reference frame for spacetime. This appears to be a general feature of background-independence, as it is also found in general relativity.

A related phenomenon is that one cannot have an instrument where all operations are (a) products of local operations, and (b) permutation-invariant (aside from the degenerate case where each of the measurement operators $M_i=N_i N_i...N_i$ act identically on every laboratory). This arises because any permutation invariant product of local measurement operators $M_{i_1}M_{i_2}...M_{i_n}$ must satisfy $M=U_gMU_g^\dag$ for all $U_g$ and therefore $M_i=M_j$ for all $i, j$. This might appear alarming: one of the fundamental tenets of the process matrix formalism is that measurements can be performed locally. However, in the next section we will discuss how a definition of locality can be recovered.

\section{Invariance with a reference frame.}

Although in the previous section we found that permutation invariance removes the distinction between points in spacetime, it is possible to recover a definition of spacetime points using a \textit{material reference frame}---a `rods and clocks' reference frame constructed out of physical systems. This corresponds to techniques used in other frameworks for relational quantum theory, such as \cite{Page1983, Rovelli1996, Poulin2006, Giovannetti2015, Miyadera2016, Giacomini2019, Smith2019quantizingtime}. In this way, we can encode all non-invariant processes, such as those typically studied in the literature, into invariant ones.

The idea of a material reference frame is to take a non-invariant process matrix, which identifies the Hilbert spaces (e.g.~$A_IA_O$, $B_IB_O$ etc.) with local laboratories labelled by some external reference frame, and add to each laboratory a physical \textit{reference} system whose quantum state encodes a label uniquely specifying that laboratory. Then, a local observer can measure this reference system to obtain information about which point of spacetime they occupy. In this way, we have encoded the information from the old abstract reference frame into a physical, observable reference frame.

This being done, we remove the external reference by making the joint process invariant under permutations. The invariant process is simply the sum of all possible permutations acting on the extended process (consisting of the system and reference frame), which gives
\begin{equation}
    W^\text{inv} = \frac{1}{n!}\mathcal{R}(W),
\end{equation}
where
\begin{equation}
	\mathcal{R}(A) := \sum_{g\in S_n}U^{SR}_gA^S[01...(n-1)]^{R_I}\mathbbm{1}^{R_O}{U^{\dag}_g}^{SR}
	\label{eq:process-inv}
\end{equation}
is a superoperator that applies to arbitrary operators (not necessarily process matrices).
In Eq.~\eqref{eq:process-inv}, the superscript on $A^S$ denotes that it is a part of the \textit{system} space $S=S^1_IS^1_O...S^n_IS^n_O$, while the superscripts $R=R_IR_O=R^1_IR^1_O...R^n_IR^n_O$ denote the \textit{reference frame} space, which contains inputs and outputs. We have used the notation $[\psi]=\ketbra{\psi}{\psi}$, so that $[01...(n-1)]^{R_I}=\ketbra{0}{0}^{R^1_I}...\ketbra{n-1}{n-1}^{R^n_I}$. Finally, $U^{SR}_g=U^S_gU^R_g$, which acts separately on the system and reference frame spaces. $U_g$ is a representation of the permutation $g$. This means that a given permutation will act on the input and output spaces of both the system and the reference frame together, so that each reference system remains associated with its corresponding laboratory, and each output space remains associated with its corresponding input space.

Essentially, we have moved from a particular process $W$ to an \textit{equivalence class} of processes (the terms in Eq.~\eqref{eq:process-inv}, related by permutations) described by $\mathcal{R}(W)$, which we treat as the fundamentally meaningful physical object, just as we consider equivalence classes of diffeomorphism-invariant spacetimes as the meaningful physical system in relativity.

Using Eq.~\eqref{eq:process-inv}, we can construct permutation-invariant processes that reproduce the statistics of arbitrary, non-invariant processes. However, as a consequence of adding a locally observable reference system to each laboratory, instruments now need to be extended so that the probability of \textit{some} measurement occurring is one. This completion turns out to be somewhat arbitrary, suggesting that there exist physically distinct instruments that are \textit{indistinguishable} when using any reference frame. 

To obtain permutation-invariant processes and measurements, we use the following maps:
\begin{align}
W\to W^\text{inv}&\equiv \frac{1}{n!}\mathcal{R}(W),\label{eq:pm-inv} \\
M_{i_1...i_n}^S&\equiv M^{S^1}_{i_1}...M^{S^n}_{i_n} \nonumber\\
\to M^\text{inv}_{i_1...i_n}&\equiv\mathcal{R}(M^{S^1}_{i_1}...M^{S^n}_{i_n}) \nonumber\\
&+\frac{1}{Nd_{O}}\left(\mathbbm{1}^{SR}-\mathcal{R}\left(\mathbbm{1}^{S}\right)\right), \label{eq:ins-inv}
\end{align}
where the instrument is composed of $N$ measurements. The additional term included in the invariant measurement operators is an arbitrary completion required in order for the instrument to be CPTP. $W^\text{inv}$ and $M^\text{inv}_{i_1...i_n}$ are now invariant under the action of $S_n$.

In the appendix, we prove that $W^\text{inv}$ are valid processes, that the $M^{\mathrm{inv}}_{i_1,...,i_n}$ are each valid elements of instruments, and that $\sum_{\{i\}}M^\mathrm{inv}_{i_1,...,i_n}$ is a CPTP map. In addition, we can show that the Born rule is maintained by the permutation invariance. Using Lemma \ref{multiplication} from the appendix,
\begin{align}
\begin{split}
&\Tr[W^\text{inv}M_i^\text{inv}]\\
&=\Tr[\frac{1}{n!}\mathcal{R}(W)(\mathcal{R}(M_i)+\frac{1}{N}(\mathbbm{1}^{SR}-\mathcal{R}(\mathbbm{1}^S))] \\
&= \Tr[\frac{1}{n!}\mathcal{R}(WM_i)]+\Tr[\frac{1}{n!}(\mathcal{R}(W)-\mathcal{R}(W))] \\
&= \Tr[\frac{1}{n!}\mathcal{R}(WM_i)]=\Tr[WM_i].
\end{split}
\end{align}

In the previous section, we mentioned that it is impossible to have an instrument in which all elements are both permutation-invariant and decompose into a product of local measurements. Since one of the central ideas of the process matrix formalism is locality, this was surprising. Here, we see that, \textit{conditionally on measuring in a particular reference frame}, the elements of a permutation-invariant instrument once more decompose into local measurements. This is because conditioning on a reference frame configuration $[R_0R_1...R_N]$ is equivalent to projecting onto $\mathbbm{1}\otimes [R_0R_1...R_N]$, which removes all terms that have a reference frame state that is different to the one being projected onto. Thus, we recover our definition of locality and see that it is only meaningful relative to a physical reference frame.

\section{Symmetry and superselection.}

Usually, symmetry constraints in quantum mechanics give rise to superselection rules on allowed states. That is, states have some `definite property' and coherences between different `types' of that property cannot exist.

The archetypal example of a superselection rule is a $U(1)$ gauge symmetry. For example, electromagnetism obeys a global $U(1)$ symmetry. This symmetry is associated with a superselection rule for electric charge: states can have any integer value of charge, but they cannot be in a superposition of two different charge values. However, it is possible to have a classical statistical mixture of positive and negative charge, such as for example when there is some classical uncertainty as to the nature of the particle being prepared such as an electron or positron. 

The reason superselection rules arise can be seen by decomposing the Hilbert space of states in terms of copies of irreducible representations of the symmetry group (in our case the group of permutations $S_n$). 

A Hilbert space with a representation $U_g$ of a group $G$ can be decomposed into a direct sum of \textit{charge sectors} $\mathcal{H}_q$, each containing an inequivalent representation of $G$. Each charge sector can in turn be decomposed into a tensor product of a \textit{gauge space} $\mathcal{M}_q$, carrying an irreducible representation of $G$, and a \textit{multiplicity space} $\mathcal{N}_q$, carrying an identity representation of $G$. The entire space decomposes as
\begin{equation}
\mathcal{H}=\bigoplus_q \mathcal{M}_q\otimes \mathcal{N}_q
,\end{equation}
so that each charge sector contains a number of copies of a particular irreducible representation. Each inequivalent representation corresponds to a different `type' of charge (in the $U(1)$ example, {the number of elementary electric charges}). In this decomposition, the \textit{twirl} - which, we recall, is the `average over all transformations' superoperator - can be expressed as \cite{bartlett:2007}
\begin{equation}
\mathcal{G}=\sum_q(\mathcal{D}_{\mathcal{M}_q}\otimes \mathcal{I}_{\mathcal{N}_q})\circ \mathcal{P}_q
\label{eq:rep-twirl}
,\end{equation}
where $\mathcal{D}$ is the completely depolarising map that sends each state to the maximally mixed state, $\mathcal{I}$ is the identity map, and $\mathcal{P}_q$ is the projector onto $\mathcal{H}_q$. Eq.~\eqref{eq:rep-twirl} tells us that linear operators that are $G$-invariant (and therefore twirl-invariant) must decompose as
\begin{equation}
A=\sum_q \frac{1}{d_{\mathcal{M}_q}}\mathbbm{1}_{\mathcal{M}_q} \otimes A_{\mathcal{N}_q}
\label{eq:twirl-inv-process}
,\end{equation}
where $d_{\mathcal{M}_q}=\text{dim}(\mathcal{M}_q)$. This restriction is a \textit{superselection rule}: requiring that allowed operators are block-diagonal in the different inequivalent representations is the same as saying that there can be no coherences between different charges. Additionally, Eq.~\eqref{eq:twirl-inv-process} gives us information about the physical degrees of freedom associated with each `type' of charge: for a charge $q$, the physical state space reduces to the invariant subspace $\mathcal{N}_q$.

So goes the typical interpretation of a superselection rule: physical objects have some well-defined charge that can be subject to classical uncertainty but not quantum indeterminacy. It turns out that, for processes, this standard interpretation fails\footnote{Strictly speaking, the concept of superselection rule does not apply directly to process matrices, because the normalisation constraint, Eq.~\eqref{eq:pm-normalised}, implies that not all vectors in the underlying Hilbert space represent valid processes. However, we retain a decomposition of the space of processes into charge sectors, with coherence possible between sectors. Symmetries should still imply loss of coherence; one could expect that this corresponds to a superselection rule.}. The reason it fails is because some inequivalent representations do not contain \textit{any} valid processes (whether there are \textit{no} representations that contain valid processes is an open question). Here, we will show that for any $n$-partite process with two-dimensional (qubit) local Hilbert spaces, the symmetric and antisymmetric representations never contain valid processes. First, we will consider the base case of a bipartite process, and then prove by induction that this will hold for a process with any number of laboratories, as long as each laboratory carries a single qubit.

A bipartite process with two-dimensional input and output spaces gives rise to a 16-dimensional Hilbert space spanned by $\ket{i}^{A_I}\ket{j}^{A_O}\ket{k}^{B_I}\ket{l}^{B_O}\equiv \ket{ijkl},\ i,j,k,l\in \{0,1\}$. Permutations of the two laboratories are obtained by the action of $S_2$, which has two elements: the identity element and the swap element $U_{AB}\ket{ijkl}=\ket{klij}$. There are two ineqivalent representations, the symmetric and antisymmetric representations (denoted $\mathcal{W}^+$ and $\mathcal{W}^-$), which are respectively spanned by
\begin{align}
	\ket{\psi_S^{(1)}}&=\ket{ijij},\ i,j\in\{0,1\}, \label{eq:sym-basis-1}\\
	\ket{\psi_S^{(2)}}&=\frac{1}{\sqrt{2}}(\ket{ijkl}+\ket{klij}),\ i,j\neq k,l \label{eq:sym-basis-2}
,\end{align}
for the `symmetric representation', and
\begin{align}
	\ket{\psi_A}=\frac{1}{\sqrt{2}}(\ket{ijkl}-\ket{klij}),\ i,j\neq k,l  \label{eq:antisym-basis}
,\end{align}
for the `anti-symmetric'. In all, the symmetric representation is 10-dimensional, and the antisymmetric is 6-dimensional. The superselection rule tells us that any physically realisable (permutation invariant) process must have the form
\begin{equation}
	W^{A_IA_OB_IB_O}=W^++W^-
	\label{eq:sym-process}
,\end{equation}
where $W^+=\sum_{i,j}\alpha_{ij}\ketbra{w^{+}_i}{w^{+}_j}$ and $W^-=\sum_{ij}\beta_{ij}\ketbra{w^{-}_i}{w^{-}_j}$, where $w^{\pm}_i$ are respectively basis elements of the symmetric $(+)$ or antisymmetric $(-)$ representations given in eqs.~\eqref{eq:sym-basis-1}-\eqref{eq:antisym-basis}.

Matrices of the form of Eq.~\eqref{eq:sym-process} will not in general be valid processes. Process matrices must also satisfy Eq. \eqref{eq:pm-normalised}. Solving algebraically for a closed-form constraint on the diagonals (which can be done with a computer algebra program such as Mathematica, or by hand) reveals that the trace of any $W^+$ or $W^-$ containing no forbidden causal terms and belonging solely to one of the two inequivalent representations must be zero. This violates the normalisation constraint of Eq.~\eqref{eq:pm-normalised}. However, there are matrices $W=W^++W^-$ where the $W^\pm$ individually contain some causally forbidden terms, but $W^++W^-$ does not, allowing $W$ to be a valid process satisfying the normalisation constraint.

We can use the result for bipartite qubit processes as the base case to show that for any number of qubit laboratories, there will be no valid processes living in the symmetric or antisymmetric representations. There are two essential parts to this argument. The first is that, given a process matrix $W$, tracing out any number of laboratories must result in a valid process. In particular, for a process $W^{S^1...S^n}$ with $n$ laboratories, $\frac{1}{d_{S^1...S^n\setminus S^iS^j}}\Tr_{S^1...S^n\setminus S^iS^j}[W^{S^1...S^n}]=\overset{\sim}{W}{}^{S^iS^j}$, where $\Tr_{S^1...S^n\setminus S^iS^j}$ denotes the trace over all laboratories except $S^i$ and $S^j$ and $d_{S^1...S^n\setminus S^iS^j}$ is the dimension of all the spaces except $S^i, S^j$ , must be a valid process. $\overset{\sim}{W}{}^{S^iS^j}$ is known as a reduced process.

The second part of the argument is that for any state living in the symmetric (antisymmetric) representation of $S_n$, any $(n-1)$-dimensional subsystem of that state will be in the symmetric (antisymmetric) representation of $S_{n-1}$. To see this, observe that we can write any $n$-dimensional state $\ket{\psi}$ as 
\begin{equation}
\ket{\psi}=\sum_j c_j \ket{\psi_j}^{S^1...S^{n-1}}\ket{j}^{S^n}
,\end{equation}
where $\ket{\psi_j}$ is an $(n-1)$-dimensional state, $c_j$ are some coefficients, and $\ket{j}$, $j=1,...,n$ is a basis state of $S^n$. Then, we have 
\begin{align}
\bra{k}^{S^n}\ket{\psi} &= \bra{k}^{S^n}\sum_j c_j\ket{\psi_j}^{S^1...S^{n-1}}\ket{j}^{S^n} \nonumber\\ &= c_k\ket{\psi_k}^{S^1...S^{n-1}}
.\end{align} 
If $\ket{\psi}$ lives in the symmetric representation, then $U_g\ket{\psi}=\ket{\psi}\forall g\in G$. In particular, this holds for all $g$ that leave the state $\ket{j}^{S^n}$ in system $S^n$ unchanged. From this, we can see that, for $g\in S_{n-1}$ and $U_g$ acting on the first $n-1$ subsystems,
\begin{align}
\begin{split}
&\qquad U_g\ket{\psi} = \ket{\psi} \\
&\Rightarrow \sum_j c_j U_g\ket{\psi_j}^{S^1...S^{n-1}}\ket{j}^{S^n} \\
&\qquad= \sum_j c_j\ket{\psi_j}^{S^1...S^{n-1}}\ket{j}^{S^n}\\
&\Rightarrow \bra{k}^{S^n}\sum_j c_j U_g\ket{\psi_j}^{S^1...S^{n-1}}\ket{j}^{S^n} \\
&\qquad= \bra{k}\sum_j c_j\ket{\psi_j}^{S^1...S^{n-1}}\ket{j}^{S^n} \\
&\Rightarrow c_k U_g\ket{\psi_k}^{S^1...S^{n-1}} = c_k\ket{\psi_k}^{S^1...S^{n-1}} \\
&\Rightarrow U_g\ket{\psi_k}^{S^1...S^{n-1}} = \ket{\psi_k}^{S^1...S^{n-1}}.
\end{split}
\end{align}
Therefore, the $\ket{\psi_j}^{S^1...S^{n-1}}$ will all be in the symmetric subspace, and $\Tr_{S^{n}}[\ketbra{\psi}{\psi}]/d_{S^{n}}$ will be a linear combination of operators on the symmetric subspace. This holds analogously for the antisymmetric subspace, where $U_g\ket{\psi}=\text{sgn}(g)\ket{\psi}$. The same result holds if we `project out' any number of subspaces. Taking the partial trace of a matrix in the (anti)symmetric subspace will therefore result in a matrix that is still in the (anti)symmetric subspace, where we define the (anti)symmetric subspace for matrices as the space of matrices that act on the (anti)symmetric subspace for states. We will equivalently say that these matrices belong to the (anti)symmetric representation.

Combining these two arguments, we see that for an $n$-partite process $W$ living in the symmetric (antisymmetric) representation of $S_n$, $\overset{\sim}{W}{}^{S^iS^j}$ must be a valid bipartite process and live in the symmetric (antisymmetric) representation for all $i,j=1,..,n, i\neq j$. But, we saw that there are no valid symmetric or antisymmetric bipartite processes, so this is a contradiction. This tells us that there are no symmetric or antisymmetric $n$-partite qubit processes. This proof generalises to any local Hilbert space dimension once one has proved the base case.

\section{Conclusion}

In this paper, we have used the process matrix formalism to show that it is possible to describe quantum causal order with background independence built in, under the assumption of a discretised spacetime. We have also seen that some properties of background independent processes have counterparts in general-relativistic background independence, e.g.~the `washing out' of spacetime and the need to construct a material reference frame to recover a definition of spacetime points.

Our results show that background independence is consistent with the principles of the process matrix formalism, including, with some reinterpretation, locality---which must be defined \textit{relative} to a reference frame. This follows from our discussion on local vs.~background independent measurements. 

We also investigated the general symmetry constraints imposed on processes by permutation invariance, and discovered that the constraint is different to the typical superselection rule: the standard interpretation is simply that physical systems must have a well-defined `charge', but for permutation-invariance not all charges correspond to physically realisable processes.  Instead, valid processes can be block-diagonal combinations of matrices that are not themselves valid processes. This implies that background independence in quantum mechanics cannot be interpreted analogously to other known symmetries of nature, and that a new interpretation may be necessary. 
Whether or not this `charge' can be taken seriously as a physical quantity is, for the moment, an open question.

Finally, our attention has focused on permutations---namely relabellings of laboratories. These can be understood as `classical' coordinate transformations, which do not change, for example, whether a particle is localised at a point or in a superposition. It has been proposed that combining quantum mechanics and general relativity requires considering more general, `quantum' coordinate transformations \cite{Hardy2016, zych:2020}. It is an interesting open question whether it is possible to extend our treatment to include such `quantum relabellings'.

\begin{acknowledgments}
This work was partially supported through an Australian Research Council (ARC) Discovery Early Career Researcher Award (DE170100712) and by the ARC Centre of Excellence for Engineered Quantum Systems (CE17010000). The University of Queensland (UQ) acknowledges the Traditional Owners and their custodianship of the lands on which UQ operates.
\end{acknowledgments}


\bibliographystyle{quantum}
\RaggedRight 
\bibliography{thesis}

\begin{widetext}
\section{Appendix}
Here, we prove some results stated in the main text.
\subsection{Some facts about the $\mathcal{R}$ map.}
It's useful to first prove two properties of the $\mathcal{R}$ map, namely that $\mathcal{R}(A)+\mathcal{R}(B)=\mathcal{R}(A+B)$ and $\mathcal{R}(A)\mathcal{R}(B)=\mathcal{R}(AB)$. For the first, we see that
\begin{lemma}
\label{addition}
$\mathcal{R}(A)+\mathcal{R}(B)=\mathcal{R}(A+B)$, for $A$ and $B$ linear operators.
\end{lemma}
\begin{proof}
\begin{align}
\begin{split}
\mathcal{R}(A)+\mathcal{R}(B)&=\sum_{g\in S_n}U^{SR}_gA^S [01...(n-1)]^{R}{U^\dag_g}^{SR} + \sum_{g\in S_n}U^{SR}_g B^S [01...(n-1)]^{R}{U^\dag_g}^{SR}\\
&=\sum_{g\in S_n}U^{SR}_g\left(A^S [01...(n-1)]^{R} + B^S [01...(n-1)]^{R}\right){U^\dag_g}^{SR} \\
&=\sum_{g\in S_n}U^{SR}_g\left(A^S+B^S\right)[01...(n-1)]^{R_i}{U^\dag_g}^{SR}=\mathcal{R}(A+B).
\end{split}
\end{align}
\end{proof}
The second takes slightly more work, but we obtain
\begin{lemma}
\label{multiplication}
$\mathcal{R}(A)\mathcal{R}(B)=\mathcal{R}(AB)$, for $A$ and $B$ linear operators.
\end{lemma}
\begin{proof}
\begin{align}
\begin{split}
	\mathcal{R}(A)\mathcal{R}(B)&=\sum_{g\in S_n}\sum_{h\in S_n}(U^{SR}_g A^S[01..(n-1)]^R{U^\dag_g}^{SR})(U^{SR}_hB^S[01...(n-1)]^R{U^\dag_h}^{SR}) \\
	&=\sum_{g\in S_n}\sum_{h\in S_n}(U_g^S A^S {U^\dag_g}^S U_h^S B^S {U^\dag_h}^S)(U_g^R[01..(n-1)]^R{U^\dag_g}^R U_h^R[01...(n-1)]^R{U^\dag_h}^R) \\
	&=\sum_{g\in S_n}\sum_{h\in S_n}(U^S_g A {U^\dag_g}^S U^S_h B {U^\dag_h}^{S})(\delta(g^{-1}h)U^R_g[01..(n-1)]{U^\dag_g}^R) \\
	&=\sum_{g\in S_n}U_g^{SR} (AB)^S[01...(n-1)]^R{U_g^\dag}^{SR} = \mathcal{R}(AB),
\end{split}
\end{align}
using $U_g[01...(n-1)]U^\dag_gU_h[01...(n-1)]U^\dag_h=[g(0)g(1)...g(n-1)][h(0)h(1)...h(n-1)]=\delta_{g(0)h(0)}...\delta_{g(n-1)h(n-1)}[g(0)...g(n-1)]=\delta(g^{-1}h)U_g[01...(n-1)]U^\dag_g$, with $\delta_{ij}$ being the kronecker delta and $\delta(g^{-1}h)$ being $1$ if $g^{-1}h$ is identity and $0$ otherwise (the delta function on $S_n$).
\end{proof}

\subsection{Permutation-invariant processes.}
For $W^\text{inv}$ to be a valid process, it must be a positive semi-definite matrix that has trace 1 when multiplied with any tensor product of CPTP maps. The positive semidefinite requirement is satisfied by Eq.~\eqref{eq:pm-inv} because both the tensor product and the sum of positive semidefinite matrices is positive semidefinite. We can also demonstrate that $W^\text{inv}$ gives normalised probabilities.

First, note that a permutation of a valid process is again a valid process. This is because $\Tr[UWU^{\dagger} M] = \Tr[W U^{\dagger}M U ]$ (where $M$ denotes the tensor product of local CP maps). Therefore, evaluating the Born rule for a permuted process is the same as evaluating it for the original process and permuted maps. If $\Tr[W  M ] = 1$ for arbitrary CPTP maps, then the same is true for permuted CPTP maps, meaning that the normalisation condition is preserved under permutations. Furthermore, convex combinations preserve process normalisation, so all we need is to prove that adding the reference frame systems to a valid process we get a valid process.

Let us then consider the operator $W^S[01...(n-1)]^{R_I}\mathbbm{1}^{R_O}$, where $W^S$ is a valid process. For arbitrary CPTP maps acting on local systems and reference frames, $M^{S^1R^1}_1 \cdots M^{S^nR^n}_n$, we have
  \begin{equation} \label{extendednorm}
  \Tr \left[W^S[01...(n-1)]^{R_I}\mathbbm{1}^{R_O}\cdot \left(M^{S^1R^1}_1 \cdots M^{S^nR^n}_n \right)\right]  = \Tr_S \left[W^S \cdot \left(\bar{M}_1^{S^1}\cdots\bar{M}_n^{S^n}\right)\right],
\end{equation}
where
\begin{equation}
   \bar{M}^{S^j}_j := \Tr_{R^j}\left[\left( [j-1]^{R^j_I}\mathbbm{1}^{R^j_O}\right)\cdot M^{S^jR^j}_j\right],\qquad j=1,\dots n.
\end{equation}
We see that each $\bar{M}^{S^j}_j$ is CPTP, which, in the Choi representation, means that $\Tr_{S^j_O}\bar{M}^{S^j}_j=\mathbbm{1}^{S_I^j}$. Indeed,
\begin{equation}
    \Tr_{S^j_O}\bar{M}^{S^j}_j = \Tr_{R^j_I}\left[ [j-1]^{R^j_I}\cdot\left( \Tr_{S^j_O R^j_O}M^{S^jR^j}_j\right)\right] = \Tr_{R^j_I}\left[ [j-1]^{R^j_I}\mathbbm{1}^{S^j_I}\right] = \mathbbm{1}^{S^j_I},
\end{equation}
where we have used the CPTP property of the extended maps, $\Tr_{S^j_O R^j_O}M^{S^jR^j}_j = \mathbbm{1}^{S^j_IR^j_I}$, and the normalisation of the reference frame state, $\Tr \, [j-1] = 1$.
As the $\bar{M}^{S^j}_j$ are CPTP, and $W^S$ is a valid process, the rhs of Eq.~\eqref{extendednorm} is equal to 1, which means the lhs is 1 too. This proves that $W^S[01...(n-1)]^{R_I}\mathbbm{1}^{R_O}$ satisfies the process normalisation condition.

\subsection{Permutation-invariant instruments.}
Here we show that the $M^\text{inv}_{i_1...i_n}$, as defined in Eq.~\eqref{eq:ins-inv}, is a valid instrument. To do this, we show that each element is a CP map, and that the trace over the combined system-reference frame output is $\Tr_{S_OR_O}[M^\text{inv}]=\mathbbm{1}^{S_IR_I}$, so that the maps sum to a CPTP map:
\begin{align}
\begin{split}
\Tr_{S_OR_O}[\sum_{i_1...i_n}M^\text{inv}_{i_1...i_n}]&=\Tr_{S_OR_O}[\sum_{i_1...i_n} \mathcal{R}(M_{i_1...i_n}^S)+\frac{1}{Nd_O}\left(\mathbbm{1}^{SR}-\mathcal{R}(\mathbbm{1}^S\right)] \\
&=\Tr_{S_OR_O}[\mathcal{R}(\sum_{i_1...i_n}M_{i_1...i_n}^S)] + \frac{1}{d_O}\Tr_{S_OR_O}[\mathbbm{1}^{SR}]-\frac{1}{d_O}\Tr_{S_OR_O}[\mathcal{R}(\mathbbm{1}^{S})] \\ 
&= \Tr_{S_OR_O}[\mathcal{R}(M^S)]+\mathbbm{1}^{S_IR_I}-\Tr_{S_OR_O}[\sum_g U_g^{SR}\mathbbm{1}^{S}[01...(n-1)]^{R_I}\mathbbm{1}^{R_O}{U_g^\dag}^{SR}] \\
&= \Tr_{S_OR_O}[\sum_g U^{SR}_g M^S[01...(n-1)]^{R_I}\mathbbm{1}^{R_O}{U^\dag_g}^{SR}]+\mathbbm{1}^{S_IR_I} \\
&\qquad\qquad-\Tr_{S_OR_O}[\sum_g U^{SR}_g \mathbbm{1}^{S}[01...(n-1)]^{R_I}\mathbbm{1}^{R_O}{U^\dag_g}^{SR}] \\
&= \Tr_{S_OR_O}[\sum_g{U^\dag_g}^{SR} U^{SR}_g M^S\otimes[01...(n-1)]^{R_I}\mathbbm{1}^{R_O}]+\mathbbm{1}^{S_IR_I} \\
& \qquad \qquad - \Tr_{S_OR_O}[\sum_g {U^\dag_g}^{SR} U^{SR}_g \mathbbm{1}^{S}\otimes[01...(n-1)]^{R_I}\mathbbm{1}^{R_O}] \\
&= \sum_{g}\Tr_{S_O}[M^S]\Tr_{R_O}[[01...(n-1)]^{R_I}\mathbbm{1}^{R_O}]+\mathbbm{1}^{S_IR_I} \\
&\qquad \qquad -\sum_g\Tr_{S_O}[\mathbbm{1}^{S}]\Tr_{R_O}[[01...(n-1)]^{R_I}\mathbbm{1}^{R_O}] \\
&= (n! d_{S_O}\mathbbm{1}^{S_I}\Tr_{R_O}[[01...(n-1)]^{R_I}\mathbbm{1}^{R_O}]-n!d_{S_O}\mathbbm{1}^{S_I}\Tr_{R_O}[[R]^{R_I}\mathbbm{1}^{R_O}])+\mathbbm{1}^{S_IR_I}\\
&=\mathbbm{1}^{S_IR_I}.
\end{split}
\end{align}
We now show that the $M_{i_1...i_n}^\mathrm{inv}$ are all positive semidefinite. Observe that
\begin{equation}
    M_{i_1...i_n}^\mathrm{inv}=\mathcal{R}(M^{S^1}_{i_1}...M^{S^n}_{i_n})+\frac{1}{Nd_{O}}\left(\mathbbm{1}^{SR}-\mathcal{R}\left(\mathbbm{1}^{S}\right)\right)
\end{equation}
is a sum of two terms. It suffices to show that both of the two terms are positive semidefinite. For the first, we see that
\begin{align}
    \mathcal{R}(M^{S^1}_{i_1}...M^{S^n}_{i_n})=\sum U_g(M^{S^1}_{i_1}...M^{S^n}_{i_n})U_g^\dag [R],
\end{align}
which, as $U_g(M^{S^1}_{i_1}...M^{S^n}_{i_n})U_g^\dag$ and $[R]$ are both positive semidefinite, is just a sum of positive semidefinite operators. Therefore, it is positive semidefinite. For the second term, note that $\mathbbm{1}^{SR}$ is simply the sum of all projectors, while $\mathcal{R}(\mathbbm{1}^S)=\sum \mathbbm{1}^{R_I}[R]^{R_O}$ is a sum containing only projectors, so that
\begin{align}
    \frac{1}{Nd_{O}}\left(\mathbbm{1}^{SR}-\mathcal{R}\left(\mathbbm{1}^{S}\right)\right)\geq0
\end{align}
because $\mathbbm{1}^{SR}$ contains all terms that appear in $\mathcal{R}(\mathbbm{1}^S)$, but the reverse is not true. Since both terms are diagonal, we can see that this results in all eigenvalues being positive or zero.
\newpage
\end{widetext}
\end{document}